\documentclass[final,12pt]{article}
\usepackage[colorlinks=true,linkcolor=blue,citecolor=magenta]{hyperref}
\usepackage{amssymb,amsmath,amsthm}
\usepackage{cite}
\usepackage{lmodern}
\usepackage{microtype}
\usepackage{enumerate}
\usepackage{multicol}
\usepackage[conditional,light,first,bottomafter]{draftcopy}
\draftcopyName{DRAFT\space\today}{130}
\draftcopySetScale{65}
\usepackage[letterpaper,hmargin=3.7cm,vmargin=3.3cm]{geometry}
\usepackage{setspace}
\makeatletter
\renewcommand{\section}{\@startsection%
{section}%
{1}%
{0em}%
{1.7em}%
{1.2em}%
{\normalfont\large\centering\bfseries}}
\renewcommand{\@seccntformat}[1]%
{\csname the#1\endcsname.\hspace{0.5em}}
\makeatother

\numberwithin{equation}{section}
\newtheorem{theorem}{Theorem}[section]
\newtheorem{proposition}{Proposition}[section]
\newtheorem{lemma}{Lemma}[section]
\newtheorem{corollary}{Corollary}[section]
\theoremstyle{definition}
\newtheorem{definition}{Definition}
\newtheorem{remark}{Remark}

\newcommand{\inner}[2]{\left\langle#1,#2\right\rangle}

\newcommand{\reals}{\mathbb{R}}
\newcommand{\nats}{\mathbb{N}}

\newcommand{\ie}{\emph{i.\,e.} }

\newcommand{\cf}{\emph{cf.} }
\newcommand{\bea}{\begin{eqnarray}}
\newcommand{\eea}{\end{eqnarray}}
\newcommand{\beao}{\begin{eqnarray*}}
\newcommand{\eeao}{\end{eqnarray*}}

\newcommand{\llb}{\left\lbrace}
\newcommand{\rrb}{\right\rbrace}
\newcommand\R{{\mathbb R}}

\newcommand\N{{\mathbb N}}

\newcommand\C{{\mathbb C}}

\renewcommand{\H}{\mathcal{H}}
\newcommand{\T}{{T}}
\newcommand{\HH}{{\mathcal{H}\oplus\mathcal{H}}}
\newcommand\ip[2]{\langle {#1},{#2} \rangle}

\newcommand\no[1]{\| {#1} \|}

\newcommand\mm[1]{{(#1)^{\perp}\oplus(#1)^{\perp}}}

\newcommand\pE[1]{{#1}_{\tiny{\odot}}}
\newcommand\Nk[2]{\pmb{\mathsf{N}}_{#1}({#2})}

\newcommand\oP[2]{{#1}\oplus {#2} }
\newcommand\oM[2]{{#1}\ominus {#2} }
\newcommand\cA[1]{\mathcal {#1}}

\newcommand\rE[1]{_{|_{#1}}}
\newcommand\vE[2]{{\begin{pmatrix}{#1}\\{#2}\end{pmatrix}}}

\DeclareMathOperator{\im}{Im}

\DeclareMathOperator{\dom}{dom}
\DeclareMathOperator{\ran }{ran }
\DeclareMathOperator{\mul}{mul}

\DeclareMathOperator{\Span}{span}

\DeclareMathOperator{\supp}{supp}

\DeclareMathOperator{\rank}{rank\,}

\begin{document}
\begin{titlepage}
\title{Perturbation theory for selfadjoint relations\\
\footnotetext{%
Mathematics Subject Classification(2010):
47A06;  % Linear relations (multivalued linear operators)
47B25; %Symmetric and selfadjoint operators (unbounded)
47A55  %Perturbation theory
}
\footnotetext{%
Keywords:
Closed linear relations;
Dissipative and Selfadjoint relations;
Weyl perturbation theory
}}
\author{%
\textbf{Josu\'e I. Rios-Cangas}
\\
%% ----- Institution --------
\small Departamento de F\'{i}sica Matem\'{a}tica\\[-1.5mm]
\small Instituto de Investigaciones en Matem\'aticas Aplicadas y en Sistemas\\[-1.5mm]
\small Universidad Nacional Aut\'onoma de M\'exico\\[-1.5mm]
\small C.P. 04510, Ciudad de M\'exico\\[-1.5mm]
\small \texttt{jottsmok@gmail.com}
\\[2mm]
\textbf{Luis O. Silva}\thanks{Supported by UNAM-DGAPA-PAPIIT IN110818 and SEP-CONACYT
CB-2015 254062. Part of this work was carried out while
  on sabbatical leave from UNAM with the support of PASPA-DGAPA-UNAM}
\\
%% ----- Institution --------
\small Department of Mathematical Sciences\\[-1.5mm]
\small University of Bath\\[-1.5mm]
\small Claverton Down, Bath BA2 7AY, U.K.\\[-1.5mm]
\small and\\[-1.5mm]
\small Departamento de F\'{i}sica Matem\'{a}tica\\[-1.5mm]
\small Instituto de Investigaciones en Matem\'aticas Aplicadas y en Sistemas\\[-1.5mm]
\small Universidad Nacional Aut\'onoma de M\'exico\\[-1.5mm]
\small C.P. 04510, Ciudad de M\'exico\\[-1.5mm]
\small \texttt{silva@iimas.unam.mx}
}
%%%%%%%%
\date{}
\maketitle
\vspace{-4mm}
\begin{center}
\begin{minipage}{5in}
  \centerline{{\bf Abstract}} \bigskip
We study Weyl-type perturbation
theorems in the context of linear closed relations. General results on
perturbations for dissipative relations are established. In the
particular case of selfadjoint relations, we study finite-rank
perturbations and carry out a detailed analysis of the corresponding
changes in the spectrum.
\end{minipage}
\end{center}
\thispagestyle{empty}
\end{titlepage}
%%%%%%%%%%%%%%%%%%%%%%%%%%%%%%
\section{\textbf{Introduction}}
\label{sec:intro}
Linear closed relations in a Hilbert space $\mathcal{H}$ are subspaces
(\ie closed linear sets) of $\H\oplus\H$. A particular realization of
a closed linear relation is the graph of a closed linear operator and,
since the operator can be identified with its graph, we consider
relations as generalizations of operators.

In this work, we study perturbation theory for relations when the
essential spectrum is preserved after the relation is submitted to
certain types of perturbations. There are various perturbation
theorems on the stability of the essential spectrum of operators (\cf
\cite[Thm.\,4.5.35]{MR0407617}, \cite[Thm.\,9.1.4]{MR1192782}) related
to the classical result on perturbations in the  selfadjoint case
by H. Weyl \cite{weyl-essential-spectrum-stability}. These theorems
are known as Weyl-type perturbations theorems. Some results of this
kind have been obtained for relations being perturbed by relations
(see \cite[Chap.\,7]{MR1631548} and
\cite{MR3398739,MR2993376,MR3255523,MR3860685})

There are several ways of extending the notions related to Weyl-type
perturbation theory from the operator setting to the one of
relations. In the first place, the essential spectrum of a relation
has to be defined.  One can use a general approach to the matter and
define the essential spectrum for relations as it is done for closed
operators in Banach spaces (see \cite[Sec.\,4.5.6]{MR0407617}, and
\cite[Chap.\,7]{MR1631548} in the relation setting). In connection
with this approach, there are various different definitions of the
essential spectrum for operators (see \cite[Chap.\,9]{MR929030}) which
can be extended to the case of relations \cite{MR3255523}. All these
notions reduce to the definition we use here (see
Definition~\ref{def:birman-definition}) in the case of selfadjoint
relations. Thus, taking into account that the main goal of this paper
is the detailed analysis of the spectrum of selfadjoint relations
under selfadjoint finite rank perturbations, we restrict ourselves to
Definition~\ref{def:birman-definition}. This definition of the
essential spectrum is used in \cite{MR3398739,MR2993376,MR3860685},
where perturbation theory for relations is treated in a way similar to
ours. The results given in this work concerning dissipative relations
and the fine-tuning perturbation analysis related to the rank of the
perturbation go beyond the results of
\cite{MR3398739,MR2993376,MR3860685}.

Apart from extending the definition of the essential spectrum, it is
necessary to generalize the concepts of relatively bounded and
relatively compact perturbations from operators to relations (see
\cite[Chap.\,7]{MR1631548} and \cite{MR3255523}). In this work, we
touch upon additive perturbations of relations only
tangentially. Instead, we approach the matter more generally by
studying the difference of the resolvents of relations when this
difference is a compact operator (for an even more general setting see
\cite{MR2481074}).

The main goal of this paper is the fine-tuning spectral analysis of
selfadjoint relations such that the difference of their resolvents is
a finite-rank operator. We develop the theory on the basis of
\cite[Chap.\,9.]{MR1192782} and extend some classical Weyl-type
perturbation results for operators to selfadjoint relations. To this
end, various results on finite-rank perturbations are obtained in
Section~\ref{sec:finitePerturbations} for a setting more general than
the selfadjoint one, namely, for dissipative relations. This is done
this way for future developments on the spectral theory of dissipative
and accretive relations, and, within the latter class, sectorial
relations (see in \cite{new-hassi} recent results on the matter).  In
the selfadjoint case, the results of
Section~\ref{sec:finitePerturbations} admit substantial refinements as
is shown in Section~\ref{sec:Perturbations}. First, we establish
Theorem~\ref{theorem-equal-essential-spectra} which is a general
result first proven in \cite[Thm.\,5.1]{MR3398739}. In contrast to
\cite[Chap.\,7]{MR1631548} and \cite{MR3255523}, this result does not
require any condition on the multivalued part of the
relations. Theorem~\ref{eqespecd} gives bounds on the shift of the
spectrum in terms of the rank of the difference of the
resolvents. Theorem~\ref{hatrels} deals with the spectra of
selfadjoint extensions of a symmetric relation with finite deficiency
indices (\cf \cite{MR2993376}). The last corollaries of
Section~\ref{sec:Perturbations} give conditions for spectral
interlacing.

Our results are of practical importance in the various theoretical
applications that the spectral theory of relations has; for instance
in the extension and spectral theories of operators
\cite{riossilva-expos} and the theory of canonical systems
\cite{MR1759823}. The last section provides examples related to the
spectral theory of operators.

\section{\textbf{On linear relations}}
\label{sec:linear-relations}
We consider a separable Hilbert space  $(\H,\, \ip{\cdot}{\cdot})$
with inner product antilinear in its left argument. Throughout this work, any linear set $T$
in $\HH$
is called a linear relation. Here, $\HH$ denotes the orthogonal sum of
two copies of the Hilbert space $\H$ (see \cite[Sec.\, 2.3]{MR1192782}).
Define the sets
\begin{align*}
 \dom \T:=\llb f\in \H\,:\ \vE fg\in T\rrb,&\quad
  \ran\T:=\llb g\in \H\,:\ \vE fg\in T\rrb,\\[1mm]
  \ker\T:=\llb f\in \H\,:\ \vE f0\in T\rrb,&\quad
  \mul\T:=\llb g\in \H\,:\ \vE 0g\in T\rrb,
\end{align*}
which turn out to be linear sets in $\H$. Moreover, if $\T$
is closed, then $\ker\T$ and $\mul\T$ are subspaces (i.\,e. closed
linear sets) of $\H$.

Given linear relations $T$ and $S$, and $\zeta\in \C$, we consider
the linear relations:
\beao
T+S&:=\llb\vE f{g+h}\ :\ \vE fg\in T,\ \ \vE fh\in S\rrb\quad
\zeta T:=\llb \vE f{\zeta g}\ :\ \vE fg\in T\rrb\\
ST&:=\llb \vE fk\ :\ \vE fg\in T,\ \ \vE gk\in S\rrb\qquad\T^{-1}:=\llb\vE gf\,:\, \vE fg\in T \rrb\,.
\eeao
We assume that the symbols $\dotplus$, $\oplus$, and $\ominus$ have
their standard meaning, i.\,e.,
\begin{equation}
  \label{eq:sets-relations-operations}
  \begin{split}
T\dotplus S&=\llb \vE {f+h}{g+k}\ :\ \vE fg\in T,\ \vE hk\in S,
\mbox{ and  }
T\cap S=\llb\vE00\rrb\rrb\,.\\
 \oP TS&=T\dotplus S\,, \mbox{ with }\, T\subset S^\perp. \\
 \oM TS &= T\cap S^\perp\,.
\end{split}
\end{equation}
The symbol $\oplus$ in this context strictly speaking differs from its
meaning in the expression $\HH$ given above. It will cause no
confusion to use the same symbol.

%The proof of the next assertion is similar in spirit to operators. %(see for instance \cite[Th. \, 3.2.3]{Birman})
%\begin{proposition}\label{prop:sumrecer}
%Let $T$ and $S$ be two closed relations such that $S$ is bounded. Then  $T+S$ is a closed relation.
%\end{proposition}

The adjoint of $T$ is defined by
\begin{align*}
 \T^*:=\llb\vE hk\in \HH\ :\ \ip kf=\ip hg,\ \ \forall \vE fg\in \T\rrb,
\end{align*}
which is a closed relation with the properties:
\begin{align}\label{poHs}
\T^*&=(-T^{-1})^{\perp},&S\subset  T&\Rightarrow T^*\subset S^*,\nonumber\\
\T^{**}&=\overline T,& (\alpha\T)^*&=\overline{\alpha}\T^*,\,\mbox{ with } \alpha\neq0\,,\\
(T^*)^{-1}&=(T^{-1})^*,&\ker \T^*&=(\ran \T)^{\perp}.\nonumber
\end{align}
From \eqref{poHs}, one obtains
\begin{align}\label{arens}
 \overline {\dom T}=\overline {\ran T^{-1}}=(\ker (\T^{-1})^{*})^\perp=(\mul T^{*})^{\perp}.
\end{align}
We call a linear relation $T$ bounded if there exists $C>0$ such that
$\no g\leq C\,\no f$, for all $\vE fg\in T$. Note that according to
this definition every bounded linear relation is a bounded linear
operator. It is worth remarking that, in the context of linear
relations, there are other ways of defining boundedness for relations
(see \cite{MR1631548}) so that a bounded relation is not necessarily
an operator.

The quasi-regular set
$\hat\rho(T)$ of the linear relation $T$ is defined by \beao \hat\rho(T):=\{\zeta \in \C\
:\ (T-\zeta I)^{-1}\mbox{ is bounded}\}\,.\eeao It is straightforward
to verify that this set is open and for every $\zeta\in\hat\rho(T)$ it
follows that $\ran (T-\zeta I)$ is closed if and only if $T$ is closed
\cite[Prop.\,2.4]{riossilva-expos}.  Furthermore, for any $\zeta\in\hat\rho(T)$,
the number \bea\label{ranTp} \eta_{\zeta}(T):=\dim [\ran (T-\zeta
I)]^{\perp}\eea is constant on each connected component of
$ \hat\rho(T)$. We call $\eta_{\zeta}(T)$ the deficiency index of $T$.
We define the deficiency space $\Nk\zeta T$ as
follows.  \bea\label{defspace} \Nk\zeta T:=\llb \vE f{\zeta f}\in
T\rrb,\,\,\, \zeta\in\C.  \eea Note that \eqref{defspace} is a linear
bounded relation which is closed if $T$ is closed. Moreover, by
\eqref{poHs} \beao \eta_{\zeta}(T)=\dim \ker (T^{*}-\overline \zeta
I)=\dim \Nk {\overline\zeta }{T^{*}}\,.\eeao If
$\eta_{\zeta}(T)=0$, then $ (T-\zeta I)^{-1}\in \cA B(\H)$,
where $\cA B(\H)$ denotes the class of all bounded operators having
the whole space $\H$ as their domain.

Define the regular set $\rho(T)$ of the linear relation $T$ by \beao
\rho(T):=\{\zeta \in \C\ :\ (T-\zeta I)^{-1}\in\cA B(\H)\}\,.\eeao
Note that if the linear relation is not closed, then the regular set
is empty. Clearly, the regular set is a subset of the quasi-regular
set and it is also open. For a relation $T$, we consider the sets
\begin{align*}
\sigma(T)&:=\C\backslash \rho(T), &\mbox{(spectrum)}\\
\hat\sigma(T)&:=\C\backslash \hat\rho(T),&\mbox{(spectral core)}\\
%\sigma_{r}(T)&=\sigma(T)\backslash \hat\sigma(T), &\mbox{(residual spectrum)}\\
\sigma_p(T)&:=\{\zeta \in \C\ :\ \ker (T-\zeta I)\neq \{0\}\},&\mbox{(point spectrum)}\\
\sigma_p^{\infty}(T)&:=\{\zeta \in \sigma_{p}(T)\ :\ \dim\ker (T-\zeta I)=\infty\},&\mbox{(point non-discrete spectrum)}\\
\sigma_c(T)&:=\{\zeta \in \C\ :\ \ran (T-\zeta I)\neq \overline{\ran (T-\zeta I)}\}.&\mbox{(continuous spectrum)}
%\sigma_{e}(T)&=\sigma_c(T)\cup\sigma_p^{\infty}(T),&\mbox{(essential spectrum)}\\
%\sigma_{d}(T)&=\hat\sigma(T)\backslash\sigma_{e}(T),&\mbox{(discrete spectrum)}
\end{align*}
As in the case of operators, one has \bea\label{eq:kerspec}
\sigma_{p}(T)\cup\sigma_{c}(T)=\hat\sigma(T)\,.\eea

For any two linear relations $T$ and $S$ in $\HH$, define the linear
relation $T_{S}$ in the Hilbert space $\mm {\mul S}$ (here $\oplus$
has the same meaning as in $\HH$) by
\bea\label{redrel} T_{S}:=T\cap\mm{\mul S}.\eea If $T$ is closed, then
$T_{S}$ is closed and if $T$ is an operator, then $T_{S}$ is an
operator. Besides
$(T_{S})^{-1}=(T^{-1})_{S}.$

It is useful to decompose a closed relation $T$ as follows $T=\oP{\pE T}{T_{\infty}}$, where
\begin{align*}
T_{\infty}&=\llb\vE 0g\in T\rrb,\\
\pE T&=\oM{T}{T_{\infty}}
\end{align*}
are closed linear relations called the multivalued part and the
operator part of $T$, respectively.

\begin{lemma}\label{lem:depolr00}
If $T$ is a closed relation such that
$\dom T\subset(\mul T)^{\perp}$, then
\begin{align*}
T-\zeta I=\oP{(T_{T}-\zeta I)}{T_{\infty}}\,.
\end{align*}
\end{lemma}
\begin{proof}
  Since the domain and the range of $\pE T$ belong to
  $(\mul T)^{\perp}$, it follows from \eqref{redrel} that
  \begin{equation}
    \label{eq:depolr10}
    T_{T}=\pE T\,.
  \end{equation}
Moreover,
  \begin{align*}%\label{eq:depolr00}
T-\zeta I=\oP{(\pE T-\zeta I)}{T_{\infty}}
\end{align*}
from which the assertion follows.
\end{proof}

%\sigma_{r}(T)&=\sigma_{r}(T_{T}),&\sigma_e(T)&=\sigma_e(T_{T}),\\
%\sigma_p(T)&=\sigma_p(T_{T}),&\sigma_d(T)&=\sigma_d(T_{T}).
% \end{remark}

\section{\textbf{Finite-dimensional perturbation of dissipative relations}}
\label{sec:finitePerturbations}

\begin{definition}
A relation $L$ is called dissipative if for every $\vE fg\in L$,
\begin{equation}
  \label{eq:dissipative-definition}
  \im\ip fg\geq0\,.
\end{equation}
If the equality in \eqref{eq:dissipative-definition} holds, then $L$
is said to be symmetric. Thus $L$ is symmetric if and only if $L\subset L^{*}$.
\end{definition}
As in \cite{riossilva-expos}, one can show that
\bea\label{eq:setdindex}\C_{-}\subset\hat\rho(L)\,,\eea for any closed
dissipative relation $L$. Thus one can consider the deficiency index
of $L$ (see \eqref{ranTp}) in the connected region $\C_{-}$ and denote it by
$\eta_{-}(L)$, \ie
\begin{equation*}
  \eta_{-}(L):=\dim\Nk{\overline
  \zeta}{L^{*}}\,,\quad \zeta\in\C_{-}\,.
\end{equation*}
If $L$ is a closed, symmetric relation, then
$\C\backslash\R\subset\hat\rho(L)$ and hence  one can consider
\begin{equation*}
  \eta_{+}(L):=\dim\Nk{\zeta}{L^{*}}\,,\quad \zeta\in\C_{-}
\end{equation*}
alongside $\eta_{-}(L)$. The index $\eta_{-}$ is an important
characteristic of a dissipative relation, while a symmetric relation is
characterized by the pair $\eta_{\pm}$.
\begin{definition}
  A dissipative relation $L$ is maximal when it is closed and
  $\eta_{-}(L)=0$.
\end{definition}
A maximal dissipative relation does not have proper dissipative
extensions. Note that maximality of a dissipative relation means that
$\C_{-}$ is in the regular set of the relation.
\begin{remark}\label{posr00}
In \cite[Lem.\,2.1]{MR3057107} (see also \cite{riossilva-expos}) it is shown that, for any dissipative relation $L$,
\begin{equation}\label{posr}
\dom L\subset(\mul L)^{\perp}.
\end{equation}
Moreover, it is proven in \cite[Thm.\,2.10]{riossilva-expos} that
\eqref{posr} yields
\begin{align}\label{depolr02}
\begin{aligned}
\sigma(L)&=\sigma(L_{L}),&\sigma_p(L)&=\sigma_p(L_{L}),\\
\hat\sigma(L)&=\hat\sigma(L_{L}),&\sigma_c(L)&=\sigma_c(L_{L})\,.
\end{aligned}
\end{align}
Furthermore, if $L$ is a closed symmetric relation in $\HH$, then by
\eqref{eq:depolr10} one obtains that $L_{L}$ is a closed symmetric
operator in $\mm {\mul L}$.
 \end{remark}

 % In this section we will work with dissipative extensions of
 % dissipative relations.
For any relation $T$ in $\cA B(\H)$, we use the notation
$\rank T:=\dim (\ran T)$.  In \cite[Thm\, 2.6.4]{MR1192782} it is
shown that $\rank T=m$ if and only if $\rank T^{*}=m$. Then
\bea\label{eq:T_of_finiterank} \dim(\oM{\H}{\ker T})=\rank T^*=\rank
T\,.  \eea

For maximal dissipative relations $A$ and $L$, and
$\zeta\in \rho(A)\cap\rho(L)$, we define
\begin{align}\label{Tdifreso}
 F:=(L-\zeta I)^{-1}-(A-\zeta I)^{-1}\in\cA B(\H)\,.
\end{align}
Note that if $\rank F=m<\infty$ for some
$\zeta\in \rho(A)\cap\rho(L)$, then the equality holds for every
$\zeta\in \rho(A)\cap\rho(L)$. Since we are mostly interested in the
rank of the operator $F$, its dependence on $\zeta$ is not indicated.

The following assertion relies on the fact that, when $A$ and $V$ are
maximal dissipative relations such that $\dom V=\H$, the relation
$A+V$ is maximal dissipative \cite[Thm.\,3.8]{riossilva-expos}.
\begin{lemma}
  Let $A$, $V$ be maximal dissipative relations such that $\dom
  V=\H$. If $L=A+V$ and $F$ is given by \eqref{Tdifreso}, then
  $\rank F\leq \rank V$.
\end{lemma}
\begin{proof}
  Take $\zeta\in\rho(L)\cap\rho(A)$ and consider $\vE f{h-k}\in F$,
  where $\vE fh\in(L-\zeta I)^{-1}$ and $\vE fk\in(A-\zeta
  I)^{-1}$. Since $V\in\mathcal{B}(\H)$, there is $t\in\H$ such that
  $\vE kt \in V$. Define the set \beao G:=\llb\vE t{h-k}\in\HH\,:\,\vE
  kt\in V\text{ and }\vE f{h-k}\in F\rrb\,.  \eeao A simple
  computation shows that $G$ is a linear relation. Let
  $\vE 0{h-k}\in G$, if $h\neq k$, then $\vE k{f+\zeta k}\in L$ and
  $\vE {h-k}{\zeta(h-k)}\in L$, whence
  $\zeta \in \sigma_{p}(L)\subset\sigma(L)$, which is impossible since
  $\zeta\in\rho(L)$.  Thus $G$ is a linear operator and therefore
\begin{equation*}
\dim\ran F=\dim \ran G\leq\dim\dom G\leq\dim\ran V.
\end{equation*}
\end{proof}
Let us move away from the case when $L$ is obtained from $A$ by an
additive perturbation and consider arbitrary maximal dissipative
relations.
\begin{lemma}\label{difexrm}
  If $A$ and $L$ are maximal dissipative extensions of a closed
  dissipative relation $S$ and $F$ is given by
  \eqref{Tdifreso}, then $\rank F\leq \eta_{-}(S)$.
\end{lemma}
\begin{proof}
  Let $\zeta\in\C_{-}$. Since $S-\zeta I\subset (A-\zeta
  I)\cap(L-\zeta I)$, one has $\ran (S-\zeta I)$ is contained in $\ker
  F$.  Hence by \eqref{eq:T_of_finiterank} one obtains
\begin{align*}
\rank F&=\dim(\oM{\H}{\ker F})\\
&\leq\dim[\oM{\H}{\ran (S-\zeta I)}]=\eta_{-}(S)\,.
\end{align*}
\end{proof}

The following statement is adapted from
\cite[Props.\,4.10,\,4.11]{riossilva-expos}. % For the second one, we recall that a
% relation $T$ is said to be regular if its quasi-regular set is the
% whole complex plane, that is, $\hat\rho(T)=\C$.

\begin{proposition}\label{mdrofsrn}
  Let $S$ be a closed symmetric relation with finite deficiency index
  $\eta_{-}(S)$. Then for any $\lambda\in
  \hat\rho(S)\cap(\C_{+}\cup\R)$ there exists a unique maximal
  dissipative extension $A$ of $S$ such that $\lambda$ is an
  eigenvalue of multiplicity at most $\eta_{-}(S)$. Furthermore, $A$ is
  selfadjoint if $\lambda\in \R$ while for $\lambda\in\C_{+}$ it
  follows that $A$ is nonselfadjoint.
\end{proposition}
% \begin{corollary}\label{lainesp0}
%   Let $A$ be a closed, regular, symmetric relation with finite deficiency index
%   $\eta_-(A)=n$. Assume that $L$ is a maximal dissipative
%   extension of $A$. Then the following holds:
%   \begin{enumerate}[(i)]
%   \item If $L$ is selfadjoint, then its spectrum consists of
%     isolated eigenvalues of multiplicity at most $n$.
%   \item If $L$ is not selfadjoint, then its spectral core consists of
%     eigenvalues of multiplicity at most $n$.
%   \item For $n=1$, every number in $\C_{+}\cup\R$ is an
%     eigenvalue of one, and only one, realization of $L$.
%   \end{enumerate}
% \end{corollary}

Now we turn to the analysis of the dimension of eigenspaces of
arbitrary maximal dissipative relations such that the difference of
their resolvents is a finite rank operator. To simplify the notation,
for $\lambda\in\C$ and $A$ being a closed dissipative relation, we put
\beao \mu_{A}(\lambda):=\dim\ker(A-\lambda I)\,.  \eeao
\begin{proposition}
\label{prop:birman-for-relations}
Let $A$ and $L$ be maximal dissipative relations such that $F$ given
by \eqref{Tdifreso} is a finite rank operator. If one defines \beao
G_{\lambda}:=\ker(A-\lambda I)\cap\ker(L-\lambda I)\,,  \eeao then
\begin{align}\label{desespd01}
  \dim[\oM{\ker(A-\lambda I)}{G_{\lambda}}]\leq \rank F,
  &\quad\dim[\oM{\ker(L-\lambda I)}{G_{\lambda}}]\leq \rank F\,,\\
\label{desespd02}
  \mu_A(\lambda)-\rank F\leq& \mu_L(\lambda)\leq \mu_A(\lambda)+\rank F\,.
\end{align}
\end{proposition}
\begin{proof}
  If we assume, for example, $\dim[\oM{\ker(A-\lambda
    I)}{G_{\lambda}}]> \rank F$, then in view of
  \eqref{eq:T_of_finiterank} there exists a nonzero element $f\in
  \ker(A-\lambda I)\cap\ker F$ such that $f\perp G_{\lambda}$. Thus
  $\vE f{\lambda f}\in A$ and $\vE f{(\lambda-\zeta)^{-1}f}\in(A-\zeta
  I)^{-1}$. Since $f\in \ker F$, one has $\vE
  f{(\lambda-\zeta)^{-1}f}\in(L-\zeta I)^{-1}$, which implies
  $\vE f{\lambda f}\in L$. Hence $f\in G_{\lambda}$ yielding a contradiction.
  Now we prove the right inequality in \eqref{desespd02}. It follows
  from \eqref{desespd01} that
\begin{align}
\begin{split}\label{desespd03}
\mu_A(\lambda)&=\dim[\oM{\ker(A-\lambda I)}{G_{\lambda}}]+\dim
G_{\lambda}\\
&\leq \rank F+ \dim G_{\lambda}
\leq \rank F +\mu_L(\lambda)\,.
\end{split}
\end{align}
To obtain the left inequality in \eqref{desespd02}, interchange the
roles of $A$ and $L$ in \eqref{desespd03}.
\end{proof}

Proposition~\ref{prop:birman-for-relations} shows that the eigenspaces
of $A$ and $L$ can differ only by a subspace of dimension at most
$\rank F$. Besides, it follows from \eqref{desespd02} that
$\sigma_p^{\infty}(A)= \sigma_p^{\infty}(L)$.
% It is to be expected that for a closed symmetric relation $S$, if
% there exists $\alpha\in\hat\rho(S)\cap\R$ then $S$ has equal
% deficiency indices and $\dim\Nk{\alpha}{S^{*}}=\eta_{\pm}(S)$.
\begin{proposition}
  If $A$ is a closed dissipative extension of a closed symmetric
  relation $S$ and $\lambda \in\hat \rho(S)$, then
  $\mu_A(\lambda)\leq \eta_{-}(S)$.
\end{proposition}
\begin{proof}
  It is clear from \eqref{eq:kerspec} and \eqref{eq:setdindex} that
  $\ker (A-\zeta I)$ can only have nontrivial elements when
  $\zeta\in\C_{+}\cup \R$. Thus, since $A\subset S^{*}$, one obtains
  that \beao \mu_A(\lambda)= \dim\Nk\lambda A\leq\dim\Nk\lambda
  {S^{*}}=\eta_{-}(S)\eeao for any $\lambda \in\hat
  \rho(S)\cap(\C_{+}\cup \R)$.
\end{proof}

\section{\textbf{Compact and finite-dimensional perturbation of
  selfadjoint relations}}
\label{sec:Perturbations}
We begin this section by stating the following
characterization of selfadjoint relations which in its operator
version is well known. Another characterization can be found in
\cite[Thm.\,2.5]{MR2993376}.

\begin{proposition}\label{equaa}
For $A$ a closed symmetric relation the following are equivalent: 
 \begin{enumerate}[{(i)}]
 \item\label{A0} $A$ is selfadjoint.
  \item\label{A00} $\eta_\pm(A)=0$.
  \item\label{A01} $\hat\rho(A)=\rho(A).$
  \item\label{A02} $\sigma (A)\subset \R.$
 \end{enumerate}
\end{proposition}
\begin{proof}\eqref{A0} $\Rightarrow$ \eqref{A00}: If
  $\zeta\in\C\backslash\R$, then $(A-\zeta I)^{-1}$ is an
  operator. Thus
\begin{align*}
 \{0\}=\mul (A-\zeta I)^{-1}=\ker (A -\zeta I)=\dom \Nk{\zeta}{A},
 \end{align*}
 whence $\dim \Nk\zeta{A^{*}}=0$.  \eqref{A00} $\Rightarrow$
 \eqref{A01}: If $\zeta\in\hat\rho(A)\setminus\R$, then, taking into
 account \eqref{ranTp}, one concludes that $\ran (A-\zeta I)=\H$ and
 then $\zeta\in\rho(A)$. Since $\hat\rho(A)$ is open and $\eta_\pm(A)$
 are constants in the connected components of $\hat\rho(A)$, if
 $\zeta\in\hat\rho(A)\cap\R $, then $\zeta\in\rho(A)$. Thus we have
 shown that $\hat\rho(A)=\rho(A)$. \eqref{A01} $\Rightarrow$
 \eqref{A02}: This is straightforward.  \eqref{A02} $\Rightarrow$
 \eqref{A0}: The hypothesis immediately implies that $\dom\Nk
 i{A^{*}}=\{0\}$. If $\vE fg$ is in $A^{*}$, then $\vE
 {g-if}{h}\in (A-iI)^{-1}$. Therefore \bea\label{oeoAsa} \vE
 {h}{g-i(f-h)}\in A\subset A^{*}.  \eea By linearity, $\vE
 {f-h}{i(f-h)}\in \Nk i{A^{*}}$. Thus, $f=h$ and \eqref{oeoAsa}
 implies that $\vE fg$ is in $A$.
  \end{proof}
\begin{definition}
 \label{def:birman-definition}
 The notion of essential spectrum for relations in the case of
 selfadjoint relations reduces to (see \cite[Sec.\,9.1.1]{MR1192782}
 for the case of operators and \cite{MR2993376} for relations) \beao
 \sigma_{e}(A):=\sigma_c(A)\cup\sigma_p^{\infty}(A)\,.  \eeao
\end{definition}
As a consequence of Proposition~\ref{equaa}-\eqref{A02} and Remark
  \ref{posr00}, if $A$ is a selfadjoint relation, then $A_{A}$ is a
  selfadjoint operator. Moreover, one verifies at once, on the basis
  of \eqref{eq:depolr10}, that $\ker (A-\zeta I)=\ker (A_{A}-\zeta
  I)$, which in turn implies that
  $\sigma_p^{\infty}(A)=\sigma_p^{\infty}(A_{A})$. Thus,
  \bea\label{eq:essA}\sigma_{e}(A)=\sigma_{e}(A_{A})\,.\eea
  
  The following definition is a generalization of the notion of singular
  sequences for selfadjoint operators.
\begin{definition}
 \label{def:singular}
  A sequence $\{u_{n}\}_{n\in\N}$ in the domain of a selfadjoint
  relation $A$ is said to be singular for $A$ at $\lambda\in\R$, if
  there exists a sequence $\llb\vE{u_{n}}{v_{n}}\rrb_{n\in\N}$ with
  elements in $A$ such that
  the following conditions are satisfied:
\begin{enumerate}[(i)]
\begin{multicols}{3}
\item $\displaystyle \inf_{n\in\N}\no{u_{n}}>0$,
\item $u_n\rightharpoonup 0$,
\item\label{con3:singular-sequence} $(v_{n}-\lambda u_{n})\rightarrow 0$\,,
\end{multicols}
\end{enumerate}
where $\rightharpoonup$ denotes weak convergence.
\end{definition}
\begin{remark}\label{re:singular-for-operator-part}
If in Definition~\ref{def:singular}, one writes  $v_n=t_n+s_n$, where
$t_n\in\ran \pE A$ and $s_n\in\mul A$, then
\begin{align*}
\no{t_n-\lambda u_n}^2&\leq \no{t_n-\lambda u_n}^2+\no{s_n}^2\\&=\no{(t_n+s_n)-\lambda u_n}^2\rightarrow0\,.
 \end{align*}
Hence,  $s_n\rightarrow0$ and the sequence $\llb\vE{u_{n}}{t_{n}}\rrb_{n\in\N}$, which has
elements in $\pE A\subset A$, satisfies
\eqref{con3:singular-sequence}. This means that if $A$ is not an
operator and $\{u_{n}\}_{n\in\nats}$ is singular, then
the sequence $\llb\vE{u_{n}}{v_{n}}\rrb_{n\in\N}$ in
Definition~\ref{def:singular} is not unique and, moreover, there are
sequences $\llb\vE{u_{n}}{w_{n}}\rrb_{n\in\N}$ in $A$ not satisfying
\eqref{con3:singular-sequence}.
\end{remark}

\begin{remark}\label{eq:sigular-equal-opertors}
  Note that Remark
  \ref{re:singular-for-operator-part} and \eqref{eq:depolr10} imply
  that there is a singular sequence for $A$ at $\lambda$ if and only
  if there is a singular sequence for $A_{A}$ at $\lambda$.
\end{remark}

The following result is known as the Weyl criterion and it can be
found for selfadjoint operators in \cite[Th.\,9.1.2]{MR1192782} (see
a more general version for operators in \cite[Chap.\,4 Sec.\,5]{MR0407617}).
\begin{proposition}\label{criweylrel}
  Let $A$ be a selfadjoint relation. The real number $\lambda$ belongs to
  $\sigma_{e}(A)$ if and only if there exists a singular sequence for
  $A$ at $\lambda$.
 \end{proposition}
\begin{proof}
The assertion follows from \eqref{eq:essA} and
Remark~\ref{eq:sigular-equal-opertors} since the
assertion is true for operators.
\end{proof}

Denote by $S_\infty(\H)\subset\cA B(\H)$ the set of compact operators
whose domain is $\H$. It is known that $V$ belongs to $S_\infty(\H)$
if and only if $V$ maps a weakly convergent sequence into a convergent
sequence (see \cite[Sec.\,2.6]{MR1192782}). The following assertion
is a Weyl-type perturbation theorem for relations.

\begin{proposition}\label{rem:TWaylaa}
  If $A$ and $V$ are selfadjoint relations such that $V\in
  S_\infty(\H)$, then $L=A+V$ is selfadjoint and
  $\sigma_e(L)=\sigma_e(A)$.
\end{proposition}
\begin{proof}
  The selfadjointness of $L$ follows from the fact that
  $(A+V)^{*}=A^{*}+V^{*}$ \cite[Prop.\, 2.2]{riossilva-expos}. For any 
  $\llb\vE{u_{n}}{v_{n}}\rrb_{n\in\N}\subset A$, there is  $\llb\vE{u_{n}}{w_{n}}\rrb_{n\in\N}\subset V$ and 
   \beao\llb\vE{u_n}{v_n+w_n}\rrb_{n\in\N}\subset L\,.\eeao
Then, if $\{u_n\}_{n\in\N}$ is singular
  for $A$ at $\lambda$, then $w_{n}\rightarrow 0$ and $[(v_{n}+w_{n})-\lambda
  u_{n}]\rightarrow 0$, which implies that  $\{u_n\}_{n\in\N}$ is singular for
  $L$ at $\lambda$ and, by Proposition \ref{criweylrel}, one has $\sigma_e(A)\subset\sigma_e(L)$. The other inclusion is
  obtained by noting that $A=L-V$.
 \end{proof}
 We now extend the previous result using the operator $F$ given in
 \eqref{Tdifreso}. An alternative proof of the following theorem is
 found in \cite[Thm.\,5.1]{MR3398739}.
\begin{theorem}\label{theorem-equal-essential-spectra}
  If $A$ and $L$ are selfadjoint relations and if $F$ belongs to
  $S_\infty(\H)$, then $\sigma_e(A)=\sigma_e(L)$.
\end{theorem}
\begin{proof}
  We only need to show that $\sigma_e(A)\subset\sigma_e(L)$.  If
  $\{u_n\}_{n\in\N}$ is singular for $A$ at $\lambda$, then by Remark
  \ref{eq:sigular-equal-opertors} there is a sequence
  $\llb\vE{u_n}{t_n}\rrb_{n\in\N}$ in $A_A$ in which
  \begin{equation}
    \label{eq:convofAa}
  (t_{n}-\lambda
  u_{n})\rightarrow 0\,.
  \end{equation}
  Note that $\vE{t_n-\zeta u_n}{u_n}\in (A-\zeta I)^{-1}$ and there
  exists $\vE {t_n-\zeta u_n}{w_n}\in (L-\zeta I)^{-1}$ for any
  $n\in\N$. A short computation shows that
  \bea\label{eq:aux-sequence-L} \llb \vE{w_n}{t_n-\zeta
    ({u_n-w_n})}\rrb_{n\in\N}\subset L\,.  \eea In view of Proposition
  \ref{criweylrel}, it only remains to prove that
  \eqref{eq:aux-sequence-L} makes $\{w_n\}_{n\in\N}$ be singular for
  $L$ at $\lambda$. First, one verifies that
  $\vE{t_n-\zeta u_n}{w_n-u_{n}}\in F$ and
  $\vE{u_n}{w_n-u_n}\in F(A_{A}-\zeta I)$.  Since $A_{A}$ and $F$ are
  operators, \beao F(A_{A}-\zeta I)=F(A_{A}-\lambda
  I)+(\lambda-\zeta)F\,.\eeao Then there exist
  $\vE {u_n}{s_n}\in F(A_{A}-\lambda I)$ and
\begin{equation}
  \label{eq:contention-in-F}
 \vE
  {u_n}{g_n}\in (\lambda-\zeta)F
\end{equation}
such that \bea\label{eq:auxeq00}
  {w_n-u_n}={s_n+g_n}\,,\qquad n\in\N\,.\eea The fact that
  $\vE{u_n}{t_n-\lambda u_n}$ is in
  $A_{A}-\lambda I$ implies
\begin{equation}
  \label{eq:contention-in-F-2}
\vE {t_n-\lambda u_n}{s_n}\in F\,.
\end{equation}
Since $F$ is a compact operator and $u_n\rightharpoonup 0$, it follows
from \eqref{eq:convofAa}, \eqref{eq:contention-in-F}, and
\eqref{eq:contention-in-F-2} that $g_n, \, s_{n}\rightarrow0$. Thus
\eqref{eq:auxeq00} implies \bea\label{eq:auxeq01}
(w_{n}-u_{n})\rightarrow 0\,.\eea The fact that $\{u_n\}_{n\in\N}$ is
singular then yields that
$w_n\rightharpoonup 0$ and $\inf_{n\in\N}\no{w_n}>0$. To conclude the
proof, observe that from
\eqref{eq:convofAa} and \eqref{eq:auxeq01}, one has \beao [t_n-\zeta
({u_n-w_n})]-\lambda w_n=(t_n-\lambda
u_n)-(\lambda-\zeta)(w_n-u_n)\rightarrow0\,.\eeao
\end{proof}

Let us turn to the study of the discrete spectrum of a selfadjoint
relation $A$. In view of Remark~\ref{posr00}, one can consider the
spectral theorem for selfadjoint operators. Let $E_{A _A}$ be the
spectral measure of $A_{A}$. It follows from \cite[Th.\,
6.1.3]{MR1192782} that
\begin{enumerate}[(1)]
 \item \label{spectr-selfadjoint-rel}$\sigma(A)=\supp E_{A_{A}}$.
 \item
   $\sigma_p(A)=\{\lambda\in \R :
   E_{A_{A}}\{\lambda\}\neq0\}$. The eigenspace corresponding to the
   eigenvalue $\lambda$ is
   $E_{ A_{A}}\{\lambda\}(\mul A)^{\perp}$.
 \item \label{cont-spectr-selfadjoint-rel}$\sigma_c(A)$ is the set of non-isolated points
   of $\sigma(A)$.
\end{enumerate}
Consider a bounded interval $\Delta$ and define
\beao\mu_A(\Delta):=\dim E_{A_{A}}(\Delta)(\mul A)^{\perp}\,.\eeao

The following assertion does not follow directly from
\cite[Thm.\,9.3.3]{MR1192782} and
\eqref{spectr-selfadjoint-rel})--\eqref{cont-spectr-selfadjoint-rel},
since the relations $L$ and $A$ could be such that $(\mul
A)^{\perp}$ and $(\mul L)^{\perp}$ do not coincide. This is
illustrated by the relations we give as examples in \eqref{selfadjoint-of-S}.

\begin{theorem}\label{eqespecd}
If $A$ and $L$ are two selfadjoint relations and $\rank F$ is finite, then
 \begin{align}\label{desespd}
  \mu_A(\Delta)-\rank F\leq \mu_L(\Delta)\leq \mu_A(\Delta)+\rank F.
 \end{align}
\end{theorem}
\begin{proof}
  Only one inequality in \eqref{desespd} needs to be proven.  If
  $\mu_L(\Delta)>\mu_A(\Delta)+\rank F$, then there exists a non-zero
  element $f\in E_{L_L}(\Delta)(\mul L)^{\perp}\cap\ker F$ such that
  $f\perp E_{A_{A}}(\Delta)(\mul A)^{\perp}$. This implies that there
  also exists
 \begin{align}\label{samelf}
  \vE fg\in (A-\zeta I)^{-1}\cap (L-\zeta I)^{-1}\,.
 \end{align}
 Due to \eqref{posr}, $\vE fg\in (L_{L}-\zeta I)^{-1}$. Observe that
 \begin{equation}
   \label{eq:trick-with-proj}
 \vE fg\in (L_{L} - \zeta I)^{-1}E_{L_{L}}(\Delta)=
 E_{L_{L}}(\Delta)(L_{L}-\zeta I)^{-1}
\end{equation}
which implies that $g\in  E_{L_{L}}(\Delta)(\mul
L)^{\perp}$. Thereupon, if one defines $\Delta:=(\alpha,\,\beta)$ and 
$\gamma:=(\alpha+\beta)/2$, then
\begin{equation*}
\vE g{f+(\zeta-\gamma)g}\in (L_{L}-\gamma
I)\,.
\end{equation*}
By the spectral theorem, one concludes
\begin{align}\label{Ebtes}
\begin{split}
\no{f+(\zeta-\gamma)g}^2&=\no{( L_{L}-\gamma I)g}^2\\
 &=\int_{|t-\gamma|<\xi}(t-\gamma)^2d\inner{g}{E_{L_{L}}(t)g}<\xi^2\no g^2\,,
\end{split}
\end{align}
where $\xi:=(\alpha-\beta)/2$. Since $f\perp E_{A_{A}}(\Delta)(\mul A)^{\perp}$, one obtains that
\beao f\in\oP{[E_{A_{A}}(\R\backslash\Delta)(\mul A)^{\perp}]}{[\mul
  A]}\,.\eeao Thus, one has the decomposition $f=f_1+f_2$, where
$f_1\in E_{A_{A}}(\R\backslash\Delta)(\mul A)^{\perp}$ and $f_2\in
\mul A$.  On the basis of the fact that $g\in (\mul A)^{\perp}$, which
follows from \eqref{samelf}, one has $\vE{f_1}g\in ( A_{A}-\zeta
I)^{-1}$. Recurring to the course of reasoning in
\eqref{eq:trick-with-proj}, one establishes as before that
\begin{equation*}
\vE g{f_1+(\zeta-\gamma)g}\in (A_{A}-\gamma I)\,.
\end{equation*}
Therefore
\begin{align*}
 \nonumber\no{f+(\zeta-\gamma)g}^2
&=\no{f_1+f_2+(\zeta-\gamma)g}^2\\ \nonumber&=\no{f_2}^2+\no{f_1+(\zeta-\gamma)g}^2\\
 \nonumber&\geq\no{f_1+(\zeta-\gamma)g}^2 
 \\  \nonumber &=\no{(A_{A}-\gamma I)g}^2\\
 &=\int_{|t-\gamma|\geq\xi}(t-\gamma)^2d\inner{g}{E_{A_{A}}(t)g}
\geq\xi^2\no g^2\,,
\end{align*}
which contradicts \eqref{Ebtes}.
\end{proof}

\begin{remark}
  \label{rem:same-essential-sa-extension}
  As a consequence of Lemma \ref{difexrm} and Theorem
  \ref{theorem-equal-essential-spectra}, if $S$ is a closed symmetric
  relation with finite and equal deficiency indices, then its
  selfadjoint extensions have the same essential spectrum.
\end{remark}
The next definition is based on  the analogous notion for
operators given in \cite[Sec.\,9.3]{MR1192782}.
\begin{definition}
  \label{def:spectral-lacuna}
  The interval $\Delta=(\gamma-\xi,\gamma+\xi)$ is a spectral lacuna
  (or simply lacuna)
  of a symmetric relation $S$ when
  \beao
 \no f\leq\frac1{\xi}\no{g-\gamma f}\quad\text{ for all }\vE fg\in S\,.
\eeao
\end{definition}
Notice that a spectral lacuna consists of quasi-regular points of $S$
and each quasi-regular point of $S$ belongs to a lacuna of $S$.

The following result is a generalization of \cite[Thm.\,9.3.6]{MR1192782}.

\begin{theorem}\label{hatrels}
  Let $A$ be a selfadjoint extension of a closed symmetric relation
  $S$. If $\Delta$ is a lacuna of $S$, then
  $\mu_A(\Delta)\leq \eta_{-}(S).$ \end{theorem}
\begin{proof}
  From \eqref{posr}, it follows that
  $\dom S\subset \dom A\subset (\mul A)^{\perp}$. Thus \beao\pE
  S\subset \oP{(\mul A)^{\perp}}{(\mul A)^{\perp}}\,.\eeao So
  $S_{A}=\pE S$ and this implies that $S_{A}$ is a closed symmetric
  operator and $A_{A}$ is its selfadjoint extension, which is also an
  operator. Since $S_{A}\subset S$, $\Delta$ is also a lacuna of
  $S_{A}$. We use \eqref{lem:depolr00} to obtain that
 \begin{align*}
 \ran(S-\zeta I)&=\oP{\ran (\pE S-\zeta I)}{\mul S}\\
 &=\oP{\ran (S_{A}-\zeta I)}{\mul S}\subset \oP{\ran (S_{A}-\zeta I)}{\mul A},
 \end{align*}
whence 
\begin{align*}
\oM{\H}{ \ran(S-\zeta I)}&\supset\oM{\H}{(\oP{\ran (S_{A}-\zeta I)}{\mul A})}\\
&=\oM{(\oM{\H}{\mul A})}{\ran (S_{A}-\zeta I)}\\&=\oM{(\mul A)^{\perp}}{\ran (S_{A}-\zeta I)}.
\end{align*}
Thus
\begin{align*}
\eta_{-}(S)&=\dim \oM{\H}{ \ran(S-\zeta I)}\\
&\leq\dim\oM{(\mul A)^{\perp}}{\ran (S_{A}-\zeta I)}\\
&=\eta_{-}(S_{A}).
\end{align*}
Using the fact that the theorem holds for operators, one concludes
that \beao\mu_A(\Delta)\leq \eta_{-}(S_{A})\leq \eta_{-}(S)\,.\eeao
\end{proof}

\begin{remark}\label{espext} Under the conditions of the previous
  result, if $\eta_{-}(S)=n<\infty$, then, for any
  $\Delta\subset\hat\rho(S)$, the spectrum of $A$ in
  $\Delta$ is discrete and its multiplicity is at most $n$. This is
  so because every closed bounded subinterval of $\Delta$ can be
  covered by a finite number of lacunae of $\hat\rho(S)$ (see
  \cite[Cor.\,5.2]{MR2993376}).
\end{remark}
\begin{corollary}\label{carespex}
  Let $S$ be a closed symmetric relation with indices
  $\eta_{-}(S)=\eta_{+}(S)=1$. Suppose that $A$ and $L$ are distinct
  selfadjoint extensions of $S$ and $\Delta\subset\hat\rho(S).$ Then
  the spectra of $A$ and $L$ in $\Delta$ are discrete, simple and
  alternating.
\end{corollary}
\begin{proof}
  By Remark \ref{espext}, the spectra of $A$ and $L$ in $\Delta$ are
  discrete and simple.  Note that $\ran F=\{0\}$ implies $A=L$, so
  $\rank F\ge 1$. On the other hand, since $\eta_{-}(S)=1$,
  it follows from Lemma~\ref{difexrm} that $\rank F\le 1$. Thus, \eqref{desespd} yields
\begin{align}\label{desc1}
 \mu_A(\Delta)-1\leq \mu_L(\Delta)\leq \mu_A(\Delta)+1\,.
 \end{align}
 Moreover, since $A\neq L$, it follows from Proposition~\ref{mdrofsrn}
 that the spectra of $A$ and $L$ have empty
 intersection in $\Delta$.  Let $\lambda_{1},\, \lambda_{2}$ be
 neighbouring eigenvalues of $A$ in $\Delta$. By \eqref{desc1}, one
 has $\mu_L((\lambda_{1},\, \lambda_{2}))\leq 1$. If there are no
 spectral points of $L$ in $[\lambda_{1},\, \lambda_{2}]$, then the
 same is true in some open interval
 $\partial\supset[\lambda_{1},\, \lambda_{2}]$. Recurring again to
 \eqref{desc1}, one arrives at
$$0=\mu_L(\partial)\geq \mu_A(\partial)-1\geq1\,,$$
which is a contradiction. Therefore the spectra of $A$ and $L$ are
alternating. \end{proof}
The next result complements Proposition \ref{mdrofsrn}. It follows
directly from Corollary \ref{carespex}.
\begin{corollary}\label{corollary:pairwise-interlaced-spectra}
  Suppose that $S$ is closed symmetric relation with indices
  $\eta_{-}(S)=\eta_{+}(S)=1$. If $\reals\subset\hat\rho(A)$, then the
  spectra of the selfadjoint extensions of $S$ are pairwise interlaced
  and consist solely of isolated eigenvalues of multiplicity one.
\end{corollary}

\section{\textbf{Examples}}
\label{sec:examples}

Let $J$ be a selfadjoint operator in a separable Hilbert space
$\H$. For a fixed nonzero $\delta \in \H$, consider the restriction
\bea\label{eq:operator-Bdelta} B_{\delta}:=J\rE{\oM{\dom J}{\Span
    \{\delta\}}}\,.  \eea The operator $B_{\delta}$ is closed,
non-densely defined and symmetric. One verifies that \beao
B_{\delta}=J\cap\left(\Span\llb\vE 0\delta\rrb\right)^{*}\,.  \eeao
Observe that $J$ and $\Span\llb\vE 0\delta\rrb$ are linearly
independent so, in view of \cite[Sec.\,, 2]{MR1430397}, one obtains
that $B_{\delta}$ has indices
$\eta_{-}(B_{\delta})=\eta_{+}(B_{\delta})=1$. Moreover, for any
$\tau\in\R\cup\{\infty\}$ there is a unique selfadjoint extension of
$B_{\delta}$ given by \bea\label{exoBop} J(\tau)=\llb \vE f{g+\tau\ip
  {\delta} f\delta} \ :\ \vE fg\in J\rrb,\,\,\, \tau\neq\infty, \eea
and \bea\label{exoBnop} J(\infty)=B_{\delta}\dotplus
\Span\llb\vE{0}{\delta}\rrb\,. \eea When $\tau$ runs through the set
$\R\cup\{\infty\}$, $J(\tau)$ runs through all selfadjoint extensions
of $B_{\delta}$ \cite[Thm.\,2.4]{MR1430397}. Also, by
\cite[Eq.\,2.2]{MR1430397},\beao%\label{taotB}
B_{\delta}^*=J\dotplus \Span\llb\vE{0}{\delta}\rrb.  \eeao

% The following result gives a characterization of the spectra of the
% relations \eqref{exoBop} and \eqref{exoBnop}. % A straightforward calculation
% shows that the operator part of \eqref{exoBnop} is
% \bea\label{eq:opertaror-part-of-selfadjoint-extension-of-B}
% \pE{J(\infty)}= \llb\vE \varphi{g-{\frac{\ip{\delta}g}{\no
%       \delta^{2}}\delta}}\,:\,\vE \varphi g\in B_{\delta}\rrb\,.  \eea

A direct consequence of Remark~\ref{rem:same-essential-sa-extension}
is the following assertion.
\begin{proposition}
  \label{prop:same-essential-example1}
  For any nonzero $\delta\in\H$, the essential spectra of all the selfadjoint extensions of
  the symmetric operator $B_{\delta}$ given in \eqref{eq:operator-Bdelta}, are equal.
\end{proposition}

Recall that a selfadjoint operator $J$ is said to be simple when there
exists $g\in\H$ such that the linear envelope of the vectors
$E_J(\partial)g$, where $E_J$ is the spectral measure of $J$ and
$\partial$ runs through all intervals of $\R$, is dense in
$\H$ (see \cite[Sec.\,69]{MR1255973}). The vector $g$ is then called a
generating element of $J$.
\begin{proposition}
  Suppose that $J$ is simple and $\delta$ is a generating element of it. If
  $\partial$ is an interval such that $\partial\cap\sigma_{e}(J)=\emptyset$, then
  $\partial\subset\hat\rho(B_{\delta})$.
\end{proposition}
\begin{proof}
  Assume $\zeta\in \partial\cap\hat\sigma(B_{\delta})$.
  Then $\zeta\in\sigma(J)$ and, since
  $\partial\cap\sigma_{e}(J)=\emptyset$,
  $\zeta\in\sigma_{d}(J)$. Moreover $\zeta \in
  \sigma_{p}(B_{\delta})$, otherwise
  $\zeta\in\sigma_{p}^{\infty}(J)\subset\sigma_{e}(J)$. Therefore \beao
  \ker(B_{\delta}-\zeta I)=E_{J}\{\zeta\}\H\,.  \eeao
  Since $\delta$ is a generating element of $J$, one has
   $\ip f\delta\neq0$ for every nonzero $f\in
   \ker(B_{\delta}-\zeta I)$. This contradicts the fact that
   $\delta\perp\dom B_{\delta}$. Therefore
   $\partial\cap\hat\sigma(B_{\delta})$ is empty which yields that 
   $\partial\subset\hat \rho(B_{\delta})$.
\end{proof}
As a consequence of the last result, if the spectrum of $J$ is purely
discrete and $\delta$ is a cyclic vector of it, then, by Corollary
\ref{corollary:pairwise-interlaced-spectra}, the spectra of the
extensions \eqref{exoBop} and \eqref{exoBnop} are pairwise interlaced
and consist solely of isolated eigenvalues of multiplicity one. Note
that this applies to $J$ being a selfadjoint Jacobi operator with
discrete spectrum and $\delta=\delta_1$, where
$\{\delta_{k}\}_{k\in\N}$ is the canonical basis in $l_2(\N)$.

Now suppose that $\{\delta_{k}\}_{k\in\N}$ is an orthonormal basis of
$\H$ and consider the restriction \beao S:=J\rE{\oM{\dom J}{\Span
    \{\delta_{1},\delta_{2}\}}}\,.  \eeao Clearly, by \eqref{exoBnop}
one has that \bea\label{selfadjoint-of-S}
J_{\delta_{1}}:=B_{\delta_{1}}\dotplus
\Span\llb\vE{0}{\delta_{1}}\rrb\,;\quad
J_{\delta_{2}}:=B_{\delta_{2}}\dotplus \Span\llb\vE{0}{\delta_{2}}\rrb
\eea are selfadjoint extensions of $S$ and they do not have a common
multivalued part. Let us show that the selfadjoint relations
\eqref{selfadjoint-of-S} have the same essential spectrum. One
computes that \beao S=J\cap\left(\Span\llb\vE 0{\delta_{1}},\vE
  0{\delta_{2}}\rrb\right)^{*}\,.  \eeao According to
\cite[Lem.\,.5.1]{MR1077947}, the indices of $S$ are
$\eta_{-}(S)=\eta_{+}(S)=2$. Therefore, by
Remark~\ref{rem:same-essential-sa-extension}, the extensions
\eqref{selfadjoint-of-S} have the same essential spectrum.

In the first example, the results of Section~\ref{sec:Perturbations}
are applied to selfadjoint extensions of the operator
$B_{\delta}$. One of these selfadjoint extensions is a relation. In
the second example, two selfadjoint extensions of the operator $S$ are
considered. These extensions are relations with different multivalued
parts.
\\[1cm]
\noindent{\bf Acknowledgments.}  This research was supported by
UNAM-DGAPA-PAPIIT IN110818 and SEP-CONACYT CB-2015 254062. Part of
this work was carried out while LOS was on sabbatical leave from UNAM
with the support of PASPA-DGAPA-UNAM.

\bibliographystyle{amsplain}

\begin{thebibliography}{99}
\def\cprime{$'$} \def\lfhook#1{\setbox0=\hbox{#1}{\ooalign{\hidewidth
  \lower1.5ex\hbox{'}\hidewidth\crcr\unhbox0}}}
\bibitem{MR1255973}
N.~I. Akhiezer and I.~M. Glazman,
\newblock {\em Theory of linear operators in {H}ilbert space},
\newblock Dover Publications, New York, 1993.

\bibitem{MR2481074}
T.~Y. Azizov, J.~Behrndt, P.~Jonas, and C.~Trunk,
\newblock \emph{Compact and finite rank perturbations of closed linear operators and
  relations in {H}ilbert spaces},
\newblock Integr. Equ. Oper. Theory \textbf{63} (2009), no. 2, 151--163.

\bibitem{MR3057107}
T.~Y. Azizov, A.~Dijksma, and G.~Wanjala,
\newblock \emph{Compressions of maximal dissipative and self-adjoint linear relations
  and of dilations},
\newblock Linear Algebra Appl. \textbf{439} (2013), no. 3, 771--792.

\bibitem{MR1192782}
M.~S. Birman and M.~Z. Solomjak,
\newblock {\em Spectral theory of selfadjoint operators in {H}ilbert space},
\newblock Mathematics and its Applications (Soviet Series), D. Reidel
  Publishing Co., Dordrecht, 1987.

\bibitem{MR1631548}
R.~Cross,
\newblock {\em Multivalued linear operators}, volume \textbf{213} of {\em Monographs and
  Textbooks in Pure and Applied Mathematics},
\newblock Marcel Dekker, New York, 1998.

\bibitem{MR1077947}
A.~Dijksma, H.~S.~V. de~Snoo, and A.~A. El~Sabbagh,
\newblock \emph{Selfadjoint extensions of regular canonical systems with {S}tieltjes
  boundary conditions},
\newblock J. Math. Anal. Appl. \textbf{152} (1990), no. 2, 546--583.

\bibitem{MR929030}
D.~E. Edmunds and W.~D. Evans,
\newblock {\em Spectral theory and differential operators},
\newblock Oxford Mathematical Monographs, The Clarendon Press, Oxford
  University Press, New York, 1987.

\bibitem{MR1430397}
S.~Hassi and H.~de~Snoo,
\newblock \emph{One-dimensional graph perturbations of selfadjoint relations},
\newblock Ann. Acad. Sci. Fenn. Math. \textbf{22} (1997), no. 1, 123--164.

\bibitem{MR1759823}
S.~Hassi, H.~de~Snoo, and H.~Winkler,
\newblock \emph{Boundary-value problems for two-dimensional canonical systems},
\newblock Integr. Equ. Oper. Theory \textbf{36} (2000), no. 4, 445--479.

\bibitem{new-hassi}
S.~Hassi, A. Sandovici, and H.~de~Snoo,
\newblock \emph{Factorized sectorial relations, their maximal
sectorial extensions, and form sums}, to appear in Banach
J. Math. Anal. (arXiv:1903.02816).

\bibitem{MR0407617}
T.~Kato,
\newblock {\em Perturbation theory for linear operators},
\newblock Springer-Verlag, Berlin, second edition, 1976.

\bibitem{riossilva-expos}
J.~I. Rios-Cangas and L.~O. Silva,
\newblock \emph{Dissipative extension theory for linear relations},
\newblock Expo. Math. {DOI}:10.1016/j.exmath.2018.10.004 (in press).

\bibitem{MR3398739}
Y.~Shi,
\newblock \emph{Stability of essential spectra of self-adjoint subspaces under
  compact perturbations},
\newblock J. Math. Anal. Appl. \textbf{433} (2016), no. 2, 832--851.

\bibitem{MR2993376}
Y.~Shi, C.~Shao, and G.~Ren,
\newblock \emph{Spectral properties of self-adjoint subspaces},
\newblock Linear Algebra Appl. \textbf{438} (2013), no. 1, 191--218.

\bibitem{weyl-essential-spectrum-stability}
H.~Weyl,
\newblock \emph{\"{U}ber beschr\"ankte quadratiche {F}ormen, deren {D}ifferenz
  vollsteig ist},
\newblock Rend. Circ. Mat. Palermo \textbf{27} (1909) 373--392.

\bibitem{MR3255523}
D.~L. Wilcox,
\newblock \emph{Essential spectra of linear relations},
\newblock Linear Algebra Appl. \textbf{462} (2014) 110--125.

\bibitem{MR3860685}
G.~Xu and Y.~Shi,
\newblock \emph{Essential spectra of self-adjoint relations under relatively compact
  perturbations},
\newblock Linear Multilinear Algebra \textbf{66} (2018), no. 12, 2438--2467.

\end{thebibliography}

\end{document}